\documentclass{article}
\usepackage{booktabs}
\usepackage{setspace}
\usepackage{xcolor}
\usepackage{morefloats}
\usepackage[T1]{fontenc}
\usepackage[activate=false,tracking,spacing,stretch=10]{microtype}
\usepackage{xspace}

\usepackage{amsfonts,amsmath,amsthm,amssymb}
\usepackage{nicefrac}
\usepackage{tikz}
\usetikzlibrary{trees,fit,positioning,patterns,shapes.geometric,calc}
\usepackage{pgfplots}
\pgfplotsset{compat=newest}
\usepackage{graphicx}

\usepackage{enumitem}
\setlist{leftmargin=*}
\usepackage{multirow, tabu, makecell, graphics}
\setitemize{noitemsep,topsep=0pt,parsep=0pt,partopsep=0pt}
\setenumerate{noitemsep,topsep=0pt,parsep=0pt,partopsep=0pt}

\usepackage{appendix}
\usepackage{subcaption}
\usepackage{algorithm}
\usepackage{algorithmic}

\theoremstyle{definition}
\newtheorem{theorem}{Theorem}
\newtheorem{remark}{Remark}
\newtheorem{corollary}{Corollary}

\newtheorem{proposition}[theorem]{Proposition}

\newtheorem{definition}[theorem]{Definition}

\usepackage{natbib}
\usepackage{hyperref}
\PassOptionsToPackage{unicode}{hyperref}
\usepackage{cleveref}
\usepackage[margin=1in]{geometry}

\title{An Ambiguous State Machine \thanks{This version: \today}}
\author{
  Matt Stephenson\thanks{Friendship Research}}
\date{}

\begin{document}
\maketitle

\begin{abstract}
We show that a replicated state machine (such as a blockchain protocol) can retain liveness in a strategic setting even while facing substantial ambiguity over certain events. This is implemented by a complementary protocol called ``Machine II'', which generates a non‐ergodic value within chosen intervals such that no limiting frequency can be observed. We show how to implement this machine and how it might be applied strategically as a mechanism for ``veiling'' actions. We demonstrate that welfare-enhancing applications for veiling exist for users belonging to a wide class of ambiguity attitudes, e.g. \citep{binmore2016minimal, gul2014expected}. Our approach is illustrated with applications to forking disputes in blockchain oracles and to Constant Function Market Makers, allowing the protocol to retain liveness without exposing their users to sure‐loss.
\end{abstract}

\section{Knowing like a State Machine}
\label{sec:smith-odds}

Ken Arrow famously defined ``state'' as  ``a description of the world so complete that, if true and known, the consequences of every action would be known''.~\citep{arrow1971essays} In economics, rational agents are often modeled--to varying degrees--as being able to know and partition this state as finely as they like for the task at hand.\footnote{Note that generalizing to full knowledge, perhaps by means of rational expectation uncertainty, tends to entail a variety of counterintuitive results, including the implication of identical prior probabilities among rational agents, no possibility of learning (as the term is commonly understood), and no-trade theorems which are poor fits with the observed data to say the least.} Here we treat state machines that face limits regarding how finely they can partition their knowledge. When these state machines are tasked with permissionlessly interacting with other agents in strategic situations, these machines may encounter adversarial agents with finer partitions of state than the machine can have. This puts such a state machine, along with any of those relying on it to service their needs, at a disadvantage.

In cases where the state machine faces extreme exploitation by these more knowledgeable agents, the machine's operators may attempt to intervene. One known intervention in extreme cases is attempting to fork or transition the state such that the exploit is ``retroactively'' mitigated in future states of the machine. Another is to attempt intervention during an attack, halting the machine to avoid further loss until the exploit can be remediated. In the tradition of Bayesian Decision Theory, the former can be thought of as akin to reneging on a bet, and the latter as refusing to bet. Both approaches have well-known drawbacks, in theory and in practice. 

We suggest an alternative approach that involves the state machine implementing a mechanism that forces any adversaries to effectively "know" only what the state machine itself "knows." Interestingly, we find that this strategy is equivalent to "refusing to bet" under some conditions, except without the requirement that the machine halt. Thus, we may able to ensure the integrity of machine without compromising liveness. 

The mechanism can also be said to be ``fair'', in that the state machine draws from a set of known possible distributions without committing to any one distribution ex ante. This is because the mechanism picks an element from a set of priors in such a way that no probability measure can be placed over that set. It is therefore committing only to what it ``knows'' and not presuming, for instance, that it is always optimal or even coherent to presume any central tendency among the different priors.

In what follows we will review a canonical setup from Decision Theory, and demonstrate where and how our approach theoretically intervenes. We will characterize our mechanism formally and practically. And then we will attend to some applications in a blockchain context.

\subsection{Refusing to Play the Game}
In early treatments of decision theory (e.g.~\cite{de1931sul,ramsey2000foundations}), it was customary to collapse one's beliefs about the probability of an event $E$ obtaining into a single value. But it's well-known since~\citet{smith1961consistency} that a rational agent need not express a single probability for an event $E$; rather, they may specify a range 
\[
  \underline{p}(E) 
  \;\;\le\;\; 
  p(E) 
  \;\;\le\;\; 
  \overline{p}(E)
\]
representing the lowest and highest probabilities she is willing to entertain for $E$. This allows an agent's beliefs to be specified as an interval compatible with partial or ``imprecise'' beliefs, \(\underline{p}(E)\,\le\,\overline{p}(E)\), whereas the classical approach would directly collapse this value into a single probability, \(\underline{p}(E)=\overline{p}(E)\) 

Nevertheless, Smith makes an ingenious argument that decision theorists may treat ``interval uncertain'' agents as if they have single probabilities within the intervals. Our mechanism can seen as denying one of the conditions Smith requires for his argument, and so we offer a brief review.

\noindent
\textbf{Forcing Inter-Interval Probabilities.}
If the agent's range is $[\underline{p}(E),\;\overline{p}(E)]$, any $\alpha$ in that interval is called a \emph{medial probability}. Agents are not compelled to bet at $\alpha$ but also cannot be contradicted by doing so. Hence in a single‐shot setting, the agent might pick \(\alpha \in [\underline{p}(E),\,\overline{p}(E)]\) at will. If the interval collapses, we recover the familiar single‐probability approach.

In order to put some structure on this medial probability, Smith introduces two important conditions. The first is a sensible coherence requirement that $\underline{p}(E)$ cannot exceed $\overline{p}(E)$. The second condition, which we aim to ``undo'', is that cases $\underline{p}(E)=0$ or $\overline{p}(E)=1$ must be ruled out. This condition expresses that while an agent may refuse to bet on some odds within the interval--they are free to choose their own medial probabilities--there must exist some finite odds at which the agent is willing to bet, and the agent is also not willing to bet at arbitrarily large odds.

With these two conditions in place, Smith finds that unless the agent acts ``as if'' they have some fixed medial probability, then they can be Dutch-booked by a slightly more complex bet. That is, if an agent forever refused to “pick” and maintain any single $\alpha$, they could be subject to a complementary‐bet construction leading to sure loss.\citep[pp.~6--16]{smith1961consistency} This selection of a fixed \emph{medial probability} for each agent formally reconciles interval‐valued beliefs with single‐bet coherence. We are back to each agent having--in effect--a single fixed value probability for all events, even if the agent only has interval valued knowledge.

But recall that this argument compels a single ``medial''
probability only if $0 < \underline{p}(E) \le \overline{p}(E) < 1$. Smith notes that without such a stipulation, much of his argument ``would require obvious modifications''. Once we allow $\underline{p}(E)=0$ \emph{or} $\overline{p}(E)=1$, no single $p$ can be forced. In the next section, we construct a non-measurable device which systematically forces such extreme $0$--$1$ ranges. As a result, we will find that an agent interacting through the device might plausibly set $\underline{p}(E)=0$ and $\overline{p}(E)=1$, or anything in between, with no contradiction or dutch-booking. 

\section{Measurable and Non-measurable Machines}
Here we will introduce two machines, Machine I and Machine II, which generate frequencies. Machine I generates measurable frequencies and can be said to provide grounds to evaluate outputs according to a fixed probability. Machine II, by contrast, has the properties of producing maximally ambiguous frequencies. The frequencies are designed to be such that an agent cannot expect convergence to any single limit, even under arbitrarily many draws from the device. There is thus no reason for the agent's posterior probability to fix a unique value in $[\,\underline{p}, \,\overline{p}]$. Having verified the properties of  ``Machine~II'', an agent conditioning on that frequency is not forced to admit a single probability value. 

\paragraph{Setup.} 
We consider a two‐step device (``machine'') which generates a 0--1 outcome:
\[
  z \;\in\;\{0,1\}.
\]
\begin{enumerate}
\item In \textbf{Step~1}, the machine draws a continuous random variable
  \(
    w \sim \mathrm{Uniform}(0,1).
  \)
\item In \textbf{Step~2}, we fix a subset $S\subseteq[0,1]$.  
  The machine sets
  \[
    z \;=\;\begin{cases}
      1 &\text{if }\,w\in S,\\
      0 &\text{if }\,w\notin S.
    \end{cases}
  \]
\end{enumerate}

An agent must choose between two actions whose payoffs depend on the eventual outcomes of repeatedly using this machine.  If $S$ is measurable, we call this \emph{Machine~I} (a single-valued ``risk'' model).  If $S$ is \emph{non‐measurable}, we call it \emph{Machine~II}, which sustains ambiguity. 

\subsection{Machine I (measurable subset)}
\begin{definition}[Machine I]
\label{def:machineI-meas}
  Let $(w_n)_{n=1}^\infty$ be an i.i.d.\ sequence of random variables,
  each $w_n \sim \mathrm{Uniform}(0,1)$.  
  Let $S \subset [0,1]$ be a Lebesgue-measurable set with measure 
  $p \in [0,1]$.
  Define the sequence $(x_n)_{n=1}^\infty$ by
  \[
    x_n 
    \;=\;
    \begin{cases}
      1 & \text{if } w_n \in S,\\
      0 & \text{otherwise}.
    \end{cases}
  \]
  By the strong law of large numbers,
  \[
    \lim_{N \to \infty} \frac{1}{N} \sum_{n=1}^N x_n 
    \;=\; p 
    \quad\text{with probability 1}.
  \]
  Hence the \emph{lower empirical frequency} and \emph{upper empirical frequency} of
  $(x_n)$ coincide a.s.\ with the single value $p$:
  \[
    \underline{p}(x) \;=\; \overline{p}(x) \;=\; p
    \quad(\text{a.s.}).
  \]
  Consequently, the corresponding \emph{lower odds} and \emph{upper odds} for the
  event $E=\{x_n=1\}$ also coincide:
  \[
    \underline{P}(E) \;=\; \overline{P}(E)
    \;=\; \frac{p}{1-p}.
  \]
  Thus, from a measure-theoretic point of view, Machine I
  yields a single, precise probability for $E$.
\end{definition}

\subsection{Machine II (Non-Measurable Subset)}
In contrast to Machine I's standard ``risk'' output, Machine II is a pathologically mixing RNG designed to resist predictability even under repeated interactions. 

\label{def:MII}
\begin{definition}[Machine II]
  \emph{Let $(x_n)_{n=1}^\infty$ be a sequence 
  taking values in $\{0,1\}$.  Define the lower empirical frequency}
  \[
    \underline{p}(x) 
    \;=\;
    \inf_{\substack{k\ge1\\m_1<m_2<\dots<m_k}}\, 
      \limsup_{n\to\infty} \;\frac{x_{m_1+n} + x_{m_2+n} + \dots + x_{m_k+n}}{k}
  \]
  \emph{and the upper empirical frequency}
  \[
    \overline{p}(x) 
    \;=\;
    \sup_{\substack{k\ge1\\m_1<m_2<\dots<m_k}}\, 
      \liminf_{n\to\infty} \;\frac{x_{m_1+n} + x_{m_2+n} + \dots + x_{m_k+n}}{k}.
  \]
  \emph{We examine the limiting fraction of 1’s in *every* shifted subsequence; \(\underline{p}(x)\) is the greatest lower limit of these fractions, while \(\overline{p}(x)\) is the least upper limit. We say $(x_n)$ is a Machine II sequence if it can produce values arbitrarily close to its bounds without converging to a single limit. i.e. for every $\epsilon>0$,}
  \[
    \limsup_{n\to\infty}
      \;\frac{x_{m_1+n} + \dots + x_{m_k+n}}{k}
    \;<\;\overline{p}(x) + \epsilon
    \quad\text{and}\quad
    \liminf_{n\to\infty}
      \;\frac{x_{m_1+n} + \dots + x_{m_k+n}}{k}
    \;>\;\underline{p}(x) - \epsilon
  \]
  holds for all sufficiently large $k$. 
 For an idealized Machine II, $\underline{p}(x)=0$ and $\overline{p}(x)=1$.
\end{definition}

That is, for any chosen \(\varepsilon > 0\), there is some subsequence that stays within \(\varepsilon\) of that fraction of 1’s over arbitrarily long stretches. In other words, the sequence never settles to a single limit, but keeps “visiting” neighborhoods of different fractions indefinitely. As we will see in Section \ref{sec:agenttype}, because one can find arbitrarily long subsequences whose frequency of 1’s stays within \(\varepsilon\) of any desired value, a wide range of agent choices can be justified ex ante on the grounds of such sequences.

In the classical style, the agent evaluating this machine will have \emph{lower odds}
\[
  \underbar{P}(E) \;=\;
  \frac{\underline{p}}{1 - \underline{p}}
\]
and \emph{upper odds}
\[
  \overline{P}(E) \;=\;
  \frac{\overline{p}}{1 - \overline{p}}.
\]

\subsection{Machine II as a ``dutch book refuser''}
\paragraph{Undoing the Conditional Bet Argument}
Recall that~\cite{smith1961consistency} showed that when $0 < \underline{p}(E) \le \overline{p}(E) < 1$, 
one can force a single ``medial'' probability via arguments of consistency against dutch books.
But without this condition, any finite odds for or against $E$ can be refused without contradiction: if
\[
  \underline{p}(E)=0
  \quad\text{and}\quad
  \overline{p}(E)=1,
\]
then the agent’s \emph{lower odds} for $E$ may be taken as $0$, 
while the \emph{upper odds} are unbounded. 
Thus an agent can consistently ``assign'' probability near $0$ (or near $1$, or anything in between) and not be shown to suffer a sure loss by refusing any finite odds for or against $E$. A state machine (or agent) implementing using Machine II can avoid a dutch book being made against it.

Losing fixed medial probabilities denies us access to the friendly functional form of ``fixed value'' Bayes Theorem calculations about pure beliefs. However, it is still possible to recover this form if the user expresses a neutral preference within the interval.\citep{binmore2016minimal} Specifically, the functional that appeared as a medial ``probability'' in Smith is now interpreted as a preferential weighting of interval probabilities. We will revisit this argument in \ref{sec:dbook} where we discuss using Machine II for pricing ambiguous assets.

\subsection{Machine II Implementation}
We will here characterize the method of implementing Machine II. We describe the procedure by which our mechanism generates ambiguity as ``veiling'' between chosen intervals. Our method follows and slightly extend the procedure of using non-ergodic draws from a Cauchy distribution as described in \citep{stecher2011generating} and \citep{pinter2022make}, and note it is most effectively run with some cryptographic protections in place, e.g. processed with threshold encrypted MPC or inside a Trusted Execution Environment (TEE). 

\begin{algorithm}[ht]
\caption{Machine~II Sequence Generation}
\label{alg:CauchyMachine}
\begin{algorithmic}[1]

\STATE \textbf{Distribution Setup:}
Let $F(x) = \frac{1}{\pi} \arctan\Bigl(\frac{x - x_0}{\gamma}\Bigr) + \frac{1}{2}$ where $x_0$ is a location parameter and $\gamma > 0$ is a scale parameter.

\vspace{1ex}

\STATE \textbf{Step 1 (Initial Seeding):} Draw $Z_{0}$ uniformly from $[0,1]$ using a secure RNG.

\vspace{1ex}

\STATE \textbf{Step 2 (First Cauchy Draw):} Sample $Z_{1}$ from $\mathrm{C}[Z_0, 1]$, a Cauchy distribution with location $Z_0$ and scale $1$.

\vspace{1ex}

\STATE \textbf{Step 3 (Iterative Non-Ergodic Updates):} Fix constants $\phi, \psi \in [0,1]$. For each $n \geq 2$:
\[
Z_n \sim \mathrm{C}\Bigl[Z_{n-1}, \phi |Z_{n-2}| + \psi\Bigr]
\]
meaning $Z_n$ is drawn from a Cauchy distribution whose location 
is $Z_{n-1}$ and whose scale parameter is $(\phi\,|Z_{n-2}| + \psi)$.
Optionally, each draw may also include new random seed bits from 
secure sources, for additional unpredictability.

\vspace{1ex}

\STATE \textbf{Output:} The sequence $\{Z_n\}_{n=0}^\infty$ resists statistical convergence, suitable for Machine~II.

\end{algorithmic}
\end{algorithm}

Because the parameters $\phi$ and $\psi$ govern how the scale of the 
distribution changes in each step (and because $X_n\sim U[0,1]$ 
enters when normalizing or shifting), 
this process is \emph{non-ergodic} and exhibits 
inconsistent sample quantiles. In other words, the fraction of times $Z_n$ lies in any given interval 
need not converge, so an outside observer cannot pin down a 
unique limiting distribution. One additional step that helps safeguard our method even against generalized learning approaches (including theoretical objects like Solomonoff inductors) is that we may charge a positive sum to observe each draw. 

\subsubsection*{From Ambiguity Generation to Ambiguity Application}
We have now defined and analyzed \emph{Machine~II}, a pathological RNG
whose outputs can resist convergence for as long as we like and thus remain
non-ergodic. In practical terms, \emph{Machine~II} offers us a
``source of perpetual ambiguity'': no agent can pin down a single
limiting frequency from its draws.

In the next section, we will leverage this property by designing
mechanisms that condition on the ambiguous frequencies generated by Machine II. Decisions, actions, outcomes, etc. which condition on these frequencies are said to be \textbf{veiled}.\footnote{Binmore uses ``Muddled Strategies'' to refer to the ambiguous choice of strategy in a Battle of the Sexes game/\citep{binmore2008rational}} We can speak of ``veiled trades'', for instance, as trades which are resolved in some way via a Machine II process (e.g. Machine II veils the final price, or who gets what, etc.) Trades may take place when such users commit to a \emph{set} of plausible outcomes consistent with \emph{Machine~II}'s ambiguous draws. Each user remains free to aggregate outcomes within that set, rather than decompose them into single-point probabilities. 

In the next section we will see why veiling can support applications that better serve users. And in a later section, we will show how it can be essential to preserve the protocol's liveness and stability.

\section{Veiling Applications}
Veiling is a feature that a state machine can implement or disregard for specific use cases. We aim here to explain why a state machine may wish to do so. Machine II's veiling procedure can be said to implement an ``ambiguous device''. As \citep{badeAmbiguityDesign2023} notes ambiguous devices may increase efficiency in many general cases, enabling ``Pareto improvements over the set of all Bayes-Nash implementable social choice functions.'' But as Bade notes, Veiling may not affect the behavior of every agent in the general case.\footnote{Agents who can alter their \emph{preferences} according to an independent coin flip, for example, are free to fix medial probabilities if they like.} Thus, because we considering a state machine which may interact with any variety of agents, we adopt the perspective of \cite{sutton2006flexibility}, who noted that under a Veiling-type procedure, ``it is difficult to find any compelling or a priori justification for imposing any single behavioural postulate on all agents.'' 

In what follows, we will make use of a flexible OHEU framework from \citep{grant2023decision} that allows us to express a wide range of ambiguity attitudes across multiple canonical models from the literature. This allows us to consider agents according to a type space, where types correspond to various parameterizations of the OHEU function. We will first aim to show veiling's use in expanding the set of ex ante reasonable actions for any of these users. Then in subsequent sections we will demonstrate this with a specific example of ambiguity values in a Constant Function Market Maker.

\subsection{Users who Need Veiling}
\label{sec:agenttype}
The goal of an application which makes use of Machine II is to serve users with beliefs and utilities whose desired actions, in at least some situations, require recourse to Veiling. We will suggest this might, in practice, be a substantial fraction of users. In a set of sister papers, we will be more specific about plausible efficiency-enhancing applications of veiling in an onchain context. But here we retain the protocol-level perspective, seeing Machine II as enabling applications and interactions among some types of users.

\textbf{Agent Types and the Type Space.}  
Let $\Theta$ be a measurable type space of agents.  
Each $\theta \in \Theta$ is an ``agent type,'' here interpreted as corresponding to a pair $(\alpha_\theta,\rho_\theta)$ which can be used to parameterize the canonical ambiguity models \citep{grant2023decision}:

\begin{equation}
   U(w,b)
   \;=\;
   \bigl[\alpha\,w^\rho \;+\; (1-\alpha)\,b^\rho\bigr]^{\tfrac{1}{\rho}}.
\label{eq:OHEU}
\end{equation}

Where $w$ and $b$ denote the best and worst case outcomes, and $\alpha$ and $\rho$ are parameters that determine the agent's elasticity of substitution for the worst and best cases. These parameters may be adjusted to model the agent's attitude toward ambiguity, with intuition provided by some useful parameterizations of this function:
\begin{itemize}
    \item Constant Marginal Rate of Substitution \citep{hurwicz1951optimality,gul2014expected}: If $\rho=1$, the relative difference between worst and best case disappears from the model. As a result $\alpha$ is the constant marginal rate of substitution between the worst and best case, as in the Hurwicz Criterion. 
    \item Maxmin preferences~\citep{gilboa1989maxmin} The special case of the Hurwicz criterion of $\alpha=1$ results in considering only the worst case, sometimes called "the utility of the paranoid." Note also that because $w\ge b$, picking $\rho=-\infty$ gives the same result but in a Leontief manner.
    \item Non-constant MRS \citep{binmore2008rational}. (At $\rho = 0$, the CES utility becomes (Binmore--Cobb--Douglas):
\[
  \lim_{\rho \to 0}\bigl[\alpha w^\rho + (1-\alpha)b^\rho\bigr]^{\frac{1}{\rho}} = w^\alpha b^{1-\alpha}.
\]
This establishes the familiar (and in consumer theory, axiomatic) condition of diminishing marginal rates of substitution, with $\alpha=1$ now a geometric weighting of the outcomes.
\end{itemize}

\medskip

Each approach has its virtues, but it's worth noting that the special case of $\rho=1, \alpha = 1/2$ has the virtue of expressing ``neutrality'' over the outcomes, and is unique in allowing functional recourse to Bayes' Theorem, albeit with a somewhat different interpretation \citep{binmore2016minimal}.\label{rem:neutral} In the subsequent section on CFMMs we will find this user sometimes presents an interesting corner case that avoids certain ``sure loss'' outcomes.

We consider the above functions, with varying parameter values, to characterize the typespace of agents facing ambiguity. Thus the goal for applications serving any such users will be to identify some set of users $\nu(\theta)$ with (possibly varying) valuation functions $(\alpha_\theta,\rho_\theta)$, and then use veiling in order to allow those users' actions $\theta \mapsto a(\theta)$ to be ex ante ``reasonable''. \citet{sutton2006flexibility} defines a ``reasonable'' action as follows:

\noindent
\textbf{Reasonable Actions.}  
Denote by $\widehat{A} \,(\subseteq A)$ the (finite) set of ``reasonable actions.'' 
We assume that there is some non-degenerate distribution of agent types from which agents are drawn at random. 
Faced with a set $\widehat{A}$ of reasonable actions, each type $\theta$ is mapped into one chosen action 
\[
  a(\theta) \;\in\; \widehat{A}.
\]
Hence there is a \emph{profile} of chosen actions across the population, determined by $\theta \mapsto a(\theta)$.

A reasonable action $a'$ is one for which there exists some 
justifying sequence $(x_{1},x_{2},\dots)$, $x_{t} \in \{0,1\}$, 
making the expected payoff of $a'$ at least as large as any other $a$ 
for all future times:
\[
  \mathbb{E}\,\pi\bigl(a' \,\bigm\vert\, p_t,p_{t+1},\dots\bigr)
  \;\;\ge\;\;
  \mathbb{E}\,\pi\bigl(a \,\bigm\vert\, p_t,p_{t+1},\dots\bigr),
  \quad \forall\, a\in\mathcal{A}.
\]

\begin{proposition}[Reasonable Actions Under Ambiguity Require Veiling]
\label{prop:M2_reasonable}
Consider an agent of type $\theta \in \Theta$ whose beliefs about an event $E$ lie in 
the interval $[\underline{p}, \overline{p}]$, with corresponding 
\emph{lower odds} and \emph{upper odds} given by
\[
   \underbar{P}(E)
   \;=\;
   \frac{\underline{p}}{1 - \underline{p}}
   \quad\text{and}\quad
   \overline{P}(E)
   \;=\;
   \frac{\overline{p}}{1 - \overline{p}}.
\]
By construction, the agent’s utility over a worst-case payoff $b$ and best-case 
payoff $w$ is given by $(\alpha_\theta,\rho_\theta)$ , where $\alpha \in [0,1]$ and $\rho$ governs substitutability between best/worst outcomes in \ref{eq:OHEU}.

Now let $\{X_{t}\}$ be an infinite sequence drawn from 
\emph{Machine~II} (Definition~\ref{def:MII}), so that it visits values 
near $\underline{p}$ and $\overline{p}$ in a non-ergodic fashion 
without converging to a single limiting frequency.  For each period $t$, 
define $p_{t}$ to be the effective probability of $E$ in that period 
(e.g.\ $p_{t} \in \{\underline{p},\overline{p}\}$).  

Because $\{X_t\}$ remains non-convergent, the agent can always maintain 
both worst-case $(\underline{p})$ and best-case $(\overline{p})$ 
scenarios as plausible. As a result, \emph{Machine~II} \textbf{expands} the set of actions that can be justified for agents $\theta \in \Theta$.
\end{proposition}

\subsection{User-friendly Veiling}
Let $\nu(\theta)$ be the fraction (or measure) of agents of type $\theta$. As application wants to identify some sufficiently large and/or profitable subset $\nu(\theta)$ whose beliefs or preferences may be ambiguous. The use of Machine II can make the actions of such a user ``reasonable'' even when they are ambiguous about the outcome. Thus offering Veiling as \emph{an option}
expands the set of justifiable actions and benefits those agents whose beliefs and utilities cannot be captured by a single ``best-guess'' probability. 

The challenge users face in making ``reasonable'' actions under ambiguous beliefs is that their ambiguity means they cannot always trade or interact a perfectly continuous way, even if they stick to intervals in which they are confident about the best and worst case outcomes. Formally, they cannot maintain ``consistency under decomposition''. What these users requires is consistency under aggregation, which is something of a mirror image and is provided by veiling. Veiling over their ambiguous interval effectively refuses the requirement for decomposition while providing consistency under aggregation. Thus they can trade without being pinned down to a ``sure loss' region, as in the earlier example of \citet{smith1961consistency}.  

In the next section we will consider the specific example of an in-protocol oracle \ref{sec:dbook} that shows the importance of keeping these two properties straight. Ir provides a specific example of how ambiguous users can engage in ``reasonable'' actions if they have recourse to Veiling. 

\section{Veiling as Fair and Efficient}
Veiling a fair procedure because it treats ambiguous information in an unbiased way. That is, one straightforward reason to use Machine II is that it can process the protocol's actual information "correctly" and does not, for instance, impose a presumed-objective central tendency on the information. Perhaps unsurprisingly, enabling the protocol to be more expressive about its information--``honest'' about what it ``knows''--can be better for users. A related observation from \citet{badeAmbiguityDesign2023} was cited earlier, noting that veiling-type procedures can enable ``Pareto improvements over the set of all Bayes-Nash implementable social choice functions''.

In this section, we will illustrate the fairness and efficiency of veiling using an example of an in-protocol oracle, where a state machine has committed to reporting a value that accords to an off-chain asset, like the U.S. Dollar.

\subsection{Fair In-Protocol Oracle Resolution}
Oracles in general require careful design in even in the best of times.\cite{eskandari2021sok} Some classic examples of thoughtful designs include Augur\citep{peterson2019augur}, Gnosis\citep{team2017gnosis}, Slinky \citep{graczyk2024slinky} and--for the off-chain case--Chainlink\citep{breidenbach2021chainlink} and Pyth.

In-protocol oracles, where oracle inputs are derived directly from the validator set, are naturally difficult to implement as well. But they are also thought to present challenges that are uniquely difficulty for a protocol to face. Such challenges are well articulated in discussions extending from a proposal by Justin Drake that the Ethereum Protocol offer direct oracle services for the currencies in the IMF's SDR bucket.\citep{drake2020enshrined} In an extensive discussion, two focal critiques emerged about the proposal:
\begin{enumerate}
    \item Gatekeeping. By enshrining some oracles and not others, the protocol would effectively privilege certain information and possibly certain applications.
    \item Irresolvability. Because the desired information can't be programmatically verified on-chain, this introduces the possibility of consensus failures due to disagreement. 
\end{enumerate}

Vitalik Buterin has an example illustrating this dilemma.\citep{buterin2023dont} He considers a case where enshrined oracles work well enough that there is pressure on the gatekeeping function, and so ``over the years, other indices get added [extending to] rates for all countries in the G20.'' And then he describes an example of irresolvablity that can emerge:

\begin{quote}
"[Suppose] Brazil has an unexpectedly severe political crisis, leading to a disputed election... Brazil has a CBDC, which splits into two forks: the (northern) BRL-N, and the (southern) BRL-S... [most] Ethereum stakers provide the ETH/BRL-S rate. Major community leaders and businesses... propose to fork the chain to only include the "good stakers" providing the ETH/BRL-N rate, and drain the other stakers' balances to near-zero. Within their social media bubble, they believe that they will clearly win. However, once the fork hits, the BRL-S side proves unexpectedly strong. What they expected to be a landslide instead proves to be pretty much a 50-50 community split. At this point, the two sides are in their two separate universes with their two chains, with no practical way of coming back together. Ethereum... ends up cleaved in half by any one of the twenty G20 member states having an unexpectedly severe internal issue."
\end{quote}

The state machine is now tasked with reporting the value for an asset BRL which has an ambiguous identity: BRL-N and BRL-S. This situation is pictured in \ref{fig:BRA-bimodal}. The quote above describes the immense difficulties faced by a protocol which asks validators to pick between the two. We can view this difficulty as tasking the protocol with breaking the asymmetry in a semi-adversarial coordination game like Battle of the Sexes as in \ref{fig:BRA-bos}.~\citep{wikipedia2025battle} Both sides would prefer to get their way, but the costs of disagreement are high (forking the chain.) The protocol's rules, specified well in advance of reporting a single value, provide no clear way to break this asymmetry. As a result, there is no pressure on the protocol to fork.

\begin{figure}
\centering
\begin{tikzpicture}
\begin{axis}[
    domain=-1:8,
    samples=100,
    axis lines=none,
    clip=false,
    width=8cm,
    height=6cm,
    ymin=0, ymax=1.2,
]
\addplot[thick,red] {exp(-((x-1)^2)/1.2) + exp(-((x-6)^2)/1.2)};
\node at (axis cs:1,1) [above] {\small BRL-N};
\node at (axis cs:6,1) [above] {\small BRL-S};
\end{axis}
\end{tikzpicture}
\caption{Plausible but Divergent Values of the Brazil CBDC currency BRL during a hypothetical civil war}
   \label{fig:BRA-bimodal}
\end{figure}
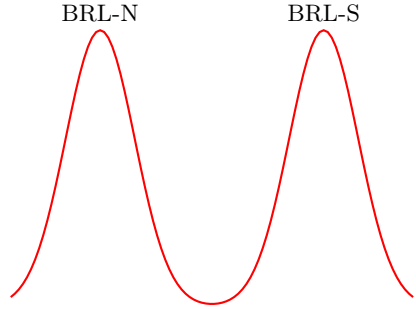

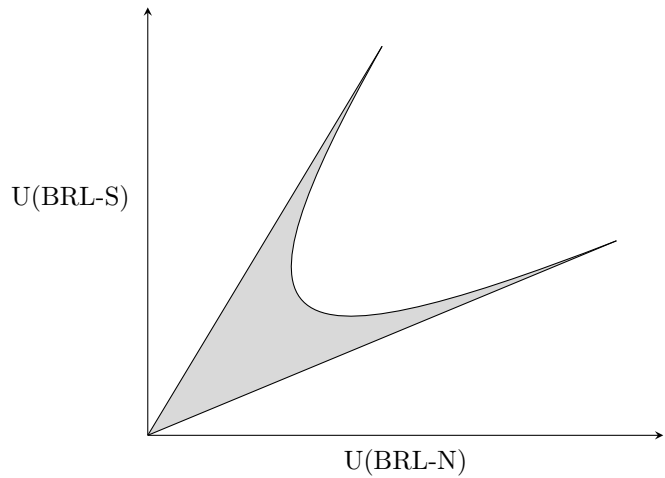
\begin{figure}
  \centering
\begin{tikzpicture}
\begin{axis}[
  axis lines=middle,
  xmin=0, xmax=2.2,
  ymin=0, ymax=2.2,
  xtick=\empty,
  ytick=\empty,
  xlabel={U(BRL-N)},
  ylabel={U(BRL-S)},
    xlabel style={
    at={(axis description cs:0.5, -0.01)},
    anchor=north
  },
  ylabel style={
    at={(axis description cs:-0.15, .5)},
    anchor=south
  },
]

\fill[gray!30]
  (axis cs:0,0) -- (axis cs:1,2)
    .. controls (axis cs:0.3,0.5) and (axis cs:0.5,0.3) ..
    (axis cs:2,1)
  -- cycle;

\draw (axis cs:0,0) -- (axis cs:1,2);
\draw (axis cs:0,0) -- (axis cs:2,1);
\draw
  (axis cs:1,2)
    .. controls (axis cs:0.3,0.5) and (axis cs:0.5,0.3) ..
    (axis cs:2,1);
\end{axis}
\end{tikzpicture}
\caption{The hypothetical on-chain BRL dispute as payoffs in an Asymmetric Coordination Game. Neither side would like to miscoordinate and fork, but each would like to get their way.)}
   \label{fig:BRA-bos}
\end{figure}

\subsection{What is the correct value for BRL?}
One obvious difficulty in the above scenario is that validators are imagined to report a single value for BRL. That suggests that one fair way to proceed is to let validators report multiple values in such a case. The BRL story makes clear that it is asking too much of an ``honest validator'' to require that they know or choose the correct side in a civil war. We might propose instead that an honest validator recognizes the situation in exactly the way Buterin \citeyear{buterin2023dont} describes: the value of BRL is ambiguous between BRL-N and BRL-S. 

But how should applications and users handle this? Honest validators had previously given a single value for BRL, but due to the civil war, they now report a set of two values (prices for BRL-N and BRL-S). One answer is that applications in such a scenario can just pick the value they like and use that. While this does push the problem onto the apps, some clever solutions exist that could enable applications to respond in a nimble way.\citep{adams2025amammauctionmanagedautomatedmarket} But even for applications that manage to coordinate on one of the values in a timely fashion, it would pose immense problems for multi-application processes. The ``same'' asset could trade at different posted prices across different applications.

Applications could also resolve to coordinate and report the ``true value'' as some specific functional to the co-domain [BRL-S, BRL-N]. For instance, applications could coordinate on using a derivative asset which is the average of the candidate values: BRL-AVG (some alternative suggestions can be found in appendix \ref{app:fairness-functions}). Lending arrangements, LP positions, and so on which were ``BRL'' would now be denominated in BRL-AVG. But even in this case where applications all coordinate on BRL-AVG, it may not resolve the actual problem for users.

\paragraph{Ambiguity about BRL-AVG} Suppose it's widely understood that once the civil war ends, the central bank will re-unify the BRL currency at some official reconciliation rate. Users may be fully confident that the central bank will eventually resolve to a value within [BRL-S,BRL-N], but they may be completely ignorant about how the bank will choose within that bracket. The bank may opt to show favoritism, be neutral, weigh other factors like inflation vs. confidence, or any number of things. Hence the bracket [BRL-S,BRL-N] may be known, but probabilities over sub-values are not.

In such a case, BRL-AVG would be understood to track the price of this eventual reconciliation. But in fact users can remain ambiguous whether BRL-S, BRL-N or any combination is in fact the correct value. These users can be systematically exposed to losing trades on BRL-AVG, for reasons that are something of a reprise of Smith's Dutch Book \ref{sec:smith-odds} from earlier. These users have a value for BRL-AVG that is not consistent under decomposition, which is what would be required of them to e.g. provide liquidity without suffering a sure loss.

Instead, these users require consistency under aggregation. That is, once a user’s bracket [BRL-S, BRL-N] is taken seriously as a \emph{set of possible values}, every combination of sub‐trades must remain within that bracket. This is provided by veiling. We call the token \textbf{BRL-Veil} that trades at a value veiled between [BRL-S, BRL-N]

\paragraph{BRL-AVG vs. BRL-Veil}In the next section, we will attend to this problem of user exploitation under BRL-AVG, using an Automated Market Maker (AMM) as a representative example. We demonstrate the following:
\begin{enumerate}
    \item Except in special cases, ambiguous users offering BRL-AVG liquidity always risk ``sure loss'' even if the price of BRL-AVG is accurate.
    \item Ambiguous users can avoid sure loss scenarios using BRL-Veil.
    \item While an ``ambiguity neutral'' user offering BRL-AVG liquidity against a single other token can avoid sure loss, we show that they can still do better with BRL-Veil than with BRL-AVG (the symmetric veiled strategy equilibrium dominates the symmetric mixed strategy for these users).
\end{enumerate}

\noindent
As a result, veiling appears to be a substantial offering in such a case.

\section{Dutch-Booking Ambiguous Users on CFMMs}
\label{sec:dbook}

In this section, we show that users who face ambiguity about the value 
of a token can be Dutch-booked (arbitraged against) when they commit to 
trades on a Constant Function Market Maker (CFMM). Recall from 
Section~\ref{sec:agenttype} that each agent type $\theta \in \Theta$ is 
characterized by a pair of aggregator parameters $(\alpha_\theta, \rho_\theta)$, 
which determine their valuation according to \eqref{eq:OHEU}. We find 
that nearly all agent types with $0 < \alpha < 1$ and $\rho \neq 0$ 
cannot provide liquidity on a CFMM without risking sure loss, 
\emph{except} for the unique ``neutral'' type with 
$\alpha=\tfrac12,\rho=1$ 
(see Remark~\ref{rem:linear-corner-case} and 
Section~\ref{sec:BoS-efficiency}).\footnote{Cases where 
$\alpha \in \{0,1\}$ are trivial since they clearly indicate a 
preference for only selling (buying) BRL-AVG unless BRL-N = BRL-S.}

Throughout, we set aside general adverse selection concerns, e.g.\ 
Loss vs.\ Rebalancing effects \citep{milionis2024automatedmarketmakinglossversusrebalancing}, 
and allow that users (and the protocol) can have perfect information 
about BRL-N and BRL-S (and thus BRL-AVG). Our finding is driven 
by the ambiguity about values within the price interval 
$[\text{BRL-S}, \,\text{BRL-N}]$,\footnote{For instance, they may trust 
the price oracles for BRL-N and BRL-S separately but have no sure 
beliefs between the two prices.} which ultimately leads to a sure-loss 
when forced to pick a \emph{single} probability.

\subsection{CFMM Setup}

\citet{angeris2023geometry} define a CFMM 
according to a concave, nondecreasing function 
$\phi : \mathbb{R}^n_{+} \to \mathbb{R}$ and a constant $k>0$. The 
\emph{reserve} vector $\mathbf{R} \in \mathbb{R}^n_{+}$ must satisfy 
$\phi(\mathbf{R}) \ge k$. In a two-token case, for example, $\phi$ 
might be $x \cdot y$ (Uniswap V2) or any generalized constant function 
in higher dimensions. Trades move the reserves within the feasible 
region $\{\mathbf{R}\colon \phi(\mathbf{R}) \ge k\}$. 
They show the condition to avoid sure loss is 
exactly the existence of a single linear (price) functional 
$p(\,\cdot\,)$ that consistently supports this feasible set. But outside of trivial and corner cases, ambiguous users do not have a single such functional, instead valuing trades with:
\begin{equation}
\label{eq:agg-rho-mini}
  U_{\alpha,\rho}(w,b)
  \;=\;
  \Bigl[\,
    \alpha\,w^\rho + (1-\alpha)\,b^\rho
  \Bigr]^{1/\rho},
  \quad
  0<\alpha<1,\;\rho\neq 0,
\end{equation}
As a result, arbitrage arises from their perspective.

\begin{proposition}[No Single Linear Functional for Ambiguous Aggregators]
\label{prop:no-linear-func}
Let $(\alpha,\rho)$ satisfy $0<\alpha<1$ and $\rho\neq0$. 
Consider the aggregator \eqref{eq:agg-rho-mini}. Then there is no 
probability vector $p$ such that 
$U_{\alpha,\rho}(X+Y) = U_{\alpha,\rho}(X) + U_{\alpha,\rho}(Y)$ 
for all two-state payoffs $X,Y \in \mathbb{R}^2$, except in the 
trivial limit $\rho\to\infty$ or $\alpha\in\{0,1\}$. Consequently, 
no single linear (price) functional can represent 
\eqref{eq:agg-rho-mini} on all two-state payoffs, implying it 
fails the no-arbitrage criterion in \cite{angeris2023geometry}.
\end{proposition}

\begin{proof}
Consider two complementary payoffs 
$X=(1,0)$ and $Y=(0,1)$. We have 
$U_{\alpha,\rho}(X) = U_{\alpha,\rho}(1,0)$ 
and 
$U_{\alpha,\rho}(X+Y) = U_{\alpha,\rho}(1,1)$. 
A linear functional would require
\[
  U_{\alpha,\rho}(1,1)
  \;=\;
  U_{\alpha,\rho}(1,0) \;+\; U_{\alpha,\rho}(0,1).
\]
Direct substitution into \eqref{eq:agg-rho-mini} shows that
$U_{\alpha,\rho}(1,1) \neq 2\,U_{\alpha,\rho}(1,0)$ 
unless $\rho\to\infty$ or $\alpha\in\{0,1\}$. Hence no single linear 
measure can represent $U_{\alpha,\rho}$ for general $0<\alpha<1$ 
and $\rho\neq0$ in a two-state setting.
\end{proof}

\begin{remark}[Neutral Users as a Special Case]
\label{rem:linear-corner-case}
A special exception arises in the two-token case if the user is linear in its bracket $[\underline{p},\overline{p}]$. This 
happens precisely at $\rho=1$ and $\alpha=\tfrac12$, such that \eqref{eq:agg-rho-mini} is evaluated:
\[
  U_{\,\tfrac12,1}(w,b)
  \;=\;
  \tfrac12\,w \;+\; \tfrac12\,b.
\]
In that corner case---\emph{and only if we consider two tokens}---the 
aggregator remains affine in $(w,b)$, and the guaranteed net-loss 
region collapses. We observed earlier (\ref{rem:neutral}) that such a user may have 
probability-like functional forms because they have a constant marginal rate of substitution between best and worst case outcomes. For these ``neutral'' users, veiling is not always 
required to avoid a sure loss. However,  
Section~\ref{sec:BoS-efficiency} proves that veiling enables them to achieve novel pareto dominant equilibria. We also note that once there are more than two possible payoffs (e.g.\ additional tokens or states), the value function for these ``neutral'' users will typically cease to be 
globally linear in all outcomes, and the sure‐loss phenomenon re‐emerges.
\end{remark}

\paragraph{Interpretation.}
Proposition~\ref{prop:no-linear-func} implies that, for nearly all 
$(\alpha,\rho)$, a CFMM that presupposes a single price can yield a guaranteed arbitrage (i.e.\ a Dutch-book) against them. This occurs precisely because no 
single linear functional can simultaneously represent all the 
user’s possible valuations in $[\text{BRL-S},\text{BRL-N}]$.

\subsection{Veiling to Avoid Sure-Loss}
\label{sec:veil}

We established in Section~\ref{sec:dbook} that agents of nearly every type $\theta \in \Theta$ cannot trade at a single linear price without risking a sure loss. We now proceed to show 
how veiling trades within a set of plausible outcomes enables them to avoid sure loss. 

\begin{proposition}[Veiling Blocks Sure Loss for Ambiguous Types]
\label{prop:no-sure-loss}
For any aggregator of the form $(\alpha,\rho)\in(0,1)\times\mathbb{R}
\setminus\{0\}$, or equivalently for any agent type $\theta\in\Theta$ 
whose aggregator lies in that class, there exists a \emph{veiled} 
strategy that avoids sure‐loss in a CFMM. Specifically, by representing 
each trade outcome as a \emph{set} of possible swaps in the bracket 
$[\mathrm{BRL\mbox{-}S}, \mathrm{BRL\mbox{-}N}]$, no single linear 
functional can force the agent’s aggregated payoff below zero in 
\emph{all} states.
\end{proposition}

\begin{proof}
Consider the set‐valued payoff representation, where the $i$th 
trade on the CFMM from reserve $\mathbf{R}$ to $\mathbf{R}'$ is modeled 
by a convex set $X^{(i)} \subseteq \mathbb{R}^n$. If each $X^{(i)}$ 
includes at least one coordinatewise nonnegative point, then the 
Minkowski sum of scaled sets
\[
  \alpha_1\,X^{(1)} \;\oplus\;\cdots\;\oplus\;\alpha_k\,X^{(k)}
\]
must contain a coordinatewise nonnegative vector in $\mathbb{R}^n$. 
Hence \emph{no} linear functional can assign strictly negative values 
to \emph{all} points in that sum; the user thus cannot be forced into 
a sure negative payoff. Concretely, every sub‐payoff in 
$[\mathrm{BRL\mbox{-}S}, \mathrm{BRL\mbox{-}N}]$ remains feasible 
without collapsing to a single price. The procedure is also 
\emph{closed under aggregation}: combining multiple veiled trades 
yields a larger set that still contains a nonnegative point. 
Thus there is no “chain” of trades that yields guaranteed loss in 
\emph{every} possible sub‐price scenario, blocking any Dutch-book 
construction.
\end{proof}

\paragraph{Maintaining the Reserve Requirement.}
A natural question is whether the CFMM’s invariant 
$\phi(\mathbf{R}) \ge k$ still holds for these set‐valued trades. 
According to \cite{angeris2023geometry}, a trade from 
$(R_1,R_2)$ to $(R_1+\Delta x, R_2-\Delta y)$ is valid only if 
$\phi(R_1+\Delta x,\, R_2-\Delta y) \ge k$. Under a veiled swap, 
the user’s final $(\Delta x,\Delta y)$ belongs to a 
\emph{set} of possible outcomes, each reflecting a different price 
in $[\mathrm{BRL\mbox{-}S},\mathrm{BRL\mbox{-}N}]$. If 
\emph{every} point in that set is individually feasible, then 
the entire set of trades complies with $\phi(\mathbf{R'}) \ge k$. 
Hence ambiguity in the price dimension does not violate the 
CFMM’s reserve constraint; it merely unifies multiple feasible 
(single‐price) trades into one set‐valued transaction.

\begin{corollary}[No Sure-Loss for All Ambiguous Agent Types]
\label{cor:no-sure-loss}
Let $\Theta$ be the collection of agent types whose parameters 
$(\alpha_\theta,\rho_\theta)$ satisfy $0<\alpha_\theta<1$ and 
$\rho_\theta\neq0$. Then, by Proposition~\ref{prop:no-sure-loss}, 
each such $\theta$ can adopt a veiled strategy to avoid Dutch-book 
exploitation on a CFMM. In other words, \emph{for all} $\theta$ 
in that broad subset of $\Theta$, there is no sequence of trades 
that forces a strictly negative payoff in every scenario.
\end{corollary}

\paragraph{Interpretation.}
Propositions~\ref{prop:no-linear-func} and \ref{prop:no-sure-loss} 
together show that although ambiguous users cannot be represented by 
a single linear measure (and thus can face sure loss if they commit 
to a single‐price mechanism), switching to \emph{veiled} (set‐valued) 
trades eliminates this vulnerability. By refusing to specify a single 
point price in $[\mathrm{BRL\mbox{-}S},\mathrm{BRL\mbox{-}N}]$ at 
each step, no linear functional can assign a negative value to all 
outcomes in the Minkowski sum. Hence, from the perspective of these 
ambiguous agent types, the CFMM now admits safe, sure‐loss‐free 
liquidity provision. In Section~\ref{sec:BoS-efficiency} we further 
argue that even the “neutral” corner case ($\alpha=\tfrac12,\rho=1$) 
may \emph{strictly gain} from veiling, thus demonstrating an 
efficiency improvement.

\subsection{Veiling is Pareto-Optimal for Neutral Users}
\label{sec:BoS-efficiency}
We now treat the special case of ``ambiguity neutral'' users described above, following \citet{binmore2016minimal} in showing how veiling enables equilibria that pareto dominate the mixed-strategy Nash in the BRL-N, BRL-S case. This can be interpreted as trade using BRL-VEIL vs. trade using BRL-AVG, implying that the former can be more efficient for these neutral users.

Recall that these users have values \(\rho=1\) and \(\alpha=\tfrac12\), and thus their aggregator is:
\[
  U_{\,\tfrac12,1}(w,b)
  ~=~
  \tfrac12\,w \;+\; \tfrac12\,b.
\] 

\begin{proposition}[Pareto-Improving Trade under Veiling]
\label{prop:pareto-veil}
As described above, consider that the currency split between BRL-S and BRL-N introduces an asymmetric coordination game where different subsets of users favor one value or the other, but also face high costs from miscoordination (e.g. forking costs). This game has payoffs:
\[
\pi_1(p,q) \;=\; p\,(1-q) \;+\; \lambda\,\bigl(1-p\bigr)\,q,
\]
with $0<\lambda<1$. We normalize the ``preferred'' payoff at 1 (where e.g. BRL-S (BRL-N) favoring users get to trade their token at the value BRL-S (BRL-N.) And thus $\lambda$ reflects the proportional disadvantage for e.g. a BRL-S user forced to trade at a BRL-N price.\footnote{Recall that the setting is one in which the protocol cannot break the asymmetry of the game by choosing one value or another. Also note that a \emph{preference} among users to resolve the value of BRL either direction (BRL-S or BRL-N) can exist even if their ``ambiguity attitude'' remains neutral, \(\rho=1\) and \(\alpha=\tfrac12\).}

The mixed-strategy Nash Equilibrium occurs at 
\[
p \;=\; q \;=\; e \;=\; \frac{1}{1+\lambda},
\]
where each player's (e.g. BRL-S and BRL-N types) have payoff:
\[
c \;=\; \pi_1(e,e) \;=\; \frac{\lambda}{1+\lambda}.
\]
Now let each player evaluate payoffs from veiled strategies (now on probability \emph{intervals} A,B) via their aggregator:
\[
U_{\tfrac12,1}(A,B) 
\;=\; 
\frac12\,\min_{(p,q)\in A\times B}\!\pi_1(p,q) 
\;+\; 
\frac12\,\max_{(p,q)\in A\times B}\!\pi_1(p,q).
\]
Then there exists a symmetric interval $[A,B]=[\,e - a^*,\,e + a^*]$
for Row and the same for Column that constitutes a setwise best response
for each, and yields a strictly higher aggregator payoff than the
single-point mixture $p=e$.  Hence the Veiled equilibrium 
$\bigl([\,e - a^*,e + a^*],[\,e - a^*,e + a^*]\bigr)$ 
\emph{Pareto-dominates} the standard point-mix $(e,e)$.
\end{proposition}

\begin{proof}[Proof Sketch]
Because $\pi_1(p,q)$ is bilinear, on any rectangle $[\underline p,\overline p]\times[\underline q,\overline q]$ both the minimum and the maximum occur at corners.  Consequently, for a symmetric interval $[e-a,e+a]$, Row’s aggregated payoff is 
\[
U_{\tfrac12,1}\bigl([e-a,e+a],[e-a,e+a]\bigr)
\;=\;
\tfrac12\,\min_{\text{corners}}\pi_1 \;+\; \tfrac12\,\max_{\text{corners}}\pi_1.
\]
A corner analysis (detailed in Appendix~\ref{app:VeilingDetails}) shows that,
\[
U_{\tfrac12,1}\bigl([e-a,e+a],[e-a,e+a]\bigr)
\;=\;
c \;+\; (1+\lambda)\,a^2
\;>\;
c.
\]
Since each player’s payoff increases by the same amount, the outcome 
$\bigl([\,e-a,e+a],[\,e-a,e+a]\bigr)$ strictly Pareto-improves upon $(e,e)$.  
Evaluating the setwise best responses to ensure no profitable expansion/contraction of the interval (see Appendix~\ref{app:VeilingDetails}), this pins down a specific $a^*$.  By symmetry, Column’s best response is likewise the same interval, completing the argument.  
\end{proof}

The above proof follows Ken Binmore's canonical example of efficiency gains under ambiguity.\citep{binmore2016minimal} We observed that the BRL-N vs. BRL-S case was an example of an asymmetric coordination game-- subsets of the protocol's users disagreed and preferred to get their way, but also faced sharp costs due to miscoordination (e.g. forking the chain.) Binmore finds a symmetric ambiguous equilibrium that pareto dominates the mixed strategy nash equilibrium in just such a case. In our example, this ambiguous equilibrium corresponds to trade at a Veiled price over the interval [BRL-S,BRL-N], and the mixed strategy equilibrium corresponds to a trade at BRL-AVG. Note also that a best response in one of these symmetric equilibria can be a mixed or ambiguous strategy. This corresponds to the idea that veiled users can trade with ambiguous users and non-ambiguous (``BRL-AVG'') alike. The striking result implies Veiling can increase efficient trade while allowing users to avoid sure-loss scenarios. The result is depicted in \ref{fig:bos-sub}.

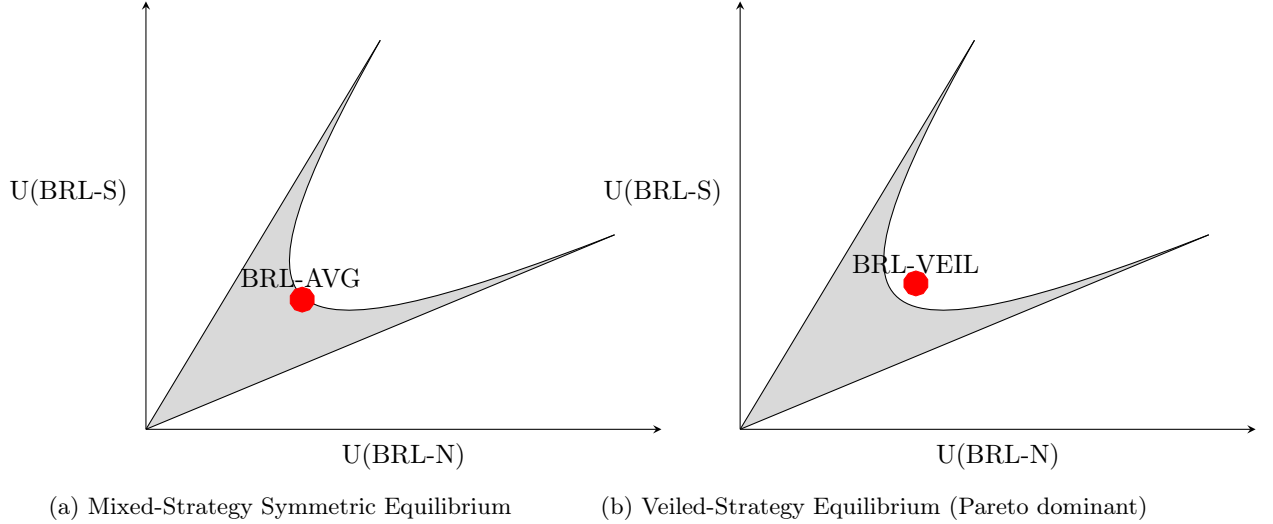
\begin{figure}[ht]
\centering

%---------------- FIRST SUBFIGURE ----------------%
\begin{subfigure}[b]{0.45\textwidth}
\centering
\begin{tikzpicture}
\begin{axis}[
  axis lines=middle,
  xmin=0, xmax=2.2,
  ymin=0, ymax=2.2,
  xtick=\empty,
  ytick=\empty,
  xlabel={U(BRL-N)},
  ylabel={U(BRL-S)},
  xlabel style={
    at={(axis description cs:0.5, -0.01)},
    anchor=north
  },
  ylabel style={
    at={(axis description cs:-0.15, .5)},
    anchor=south
  },
]

% The original BOS wedge (gray fill)
\fill[gray!30]
  (axis cs:0,0) -- (axis cs:1,2)
  .. controls (axis cs:0.3,0.5) and (axis cs:0.5,0.3) ..
  (axis cs:2,1) -- cycle;

\draw (axis cs:0,0) -- (axis cs:1,2);
\draw (axis cs:0,0) -- (axis cs:2,1);
\draw
  (axis cs:1,2)
  .. controls (axis cs:0.3,0.5) and (axis cs:0.5,0.3) ..
  (axis cs:2,1);

% Mixed-Strategy Nash star
\draw[red,fill=red] (axis cs:2/3,2/3)
  node[star,star points=5,star point ratio=1,draw=red,fill=red]{};
\node at (axis cs:0.66, 0.78) {BRL-AVG};

\end{axis}
\end{tikzpicture}
\caption{Mixed-Strategy Symmetric Equilibrium}
\end{subfigure}
\quad
%---------------- SECOND SUBFIGURE ----------------%
\begin{subfigure}[b]{0.45\textwidth}
\centering
\begin{tikzpicture}
\begin{axis}[
  axis lines=middle,
  xmin=0, xmax=2.2,
  ymin=0, ymax=2.2,
  xtick=\empty,
  ytick=\empty,
  xlabel={U(BRL-N)},
  ylabel={U(BRL-S)},
  xlabel style={
    at={(axis description cs:0.5, -0.01)},
    anchor=north
  },
  ylabel style={
    at={(axis description cs:-0.15, .5)},
    anchor=south
  },
]

% Same BOS wedge in gray
\fill[gray!30]
  (axis cs:0,0) -- (axis cs:1,2)
  .. controls (axis cs:0.3,0.5) and (axis cs:0.5,0.3) ..
  (axis cs:2,1) -- cycle;
\draw (axis cs:0,0) -- (axis cs:1,2);
\draw (axis cs:0,0) -- (axis cs:2,1);
\draw
  (axis cs:1,2)
  .. controls (axis cs:0.3,0.5) and (axis cs:0.5,0.3) ..
  (axis cs:2,1);

% Veiled-Strategy Nash star
\draw[red,fill=red] (axis cs:0.75,0.75)
  node[star,star points=5,star point ratio=1,draw=red,fill=red]{};
\node at (axis cs:0.75, 0.85) {BRL-VEIL};

\end{axis}
\end{tikzpicture}
\caption{Veiled-Strategy Equilibrium (Pareto dominant)}
\end{subfigure}

\caption{Comparison of mixed vs.\ veiled strategies in payoff space. The red dots correspond to equilibrium payoffs in the classic finding that a veiled strategy equilibrium Pareto dominates the classical mixed strategy Nash.}
\label{fig:bos-sub}
\end{figure}

\subsection{Welfare Improvement from Veiling}
Veiling has been shown to be fair and efficient for ambiguous users in a CFMM setting, helping them avoid sure-loss scenarios and provide novel pareto improving equilibria. However, the plausible applications of veiling extend well-beyond the oracle-CFMM setting. Elsewhere we consider the role of veiling in onchain applications and find cases where efficiency gains can be found including, liquidations, slippage, transaction ordering, data routing, and even MEV auctions comprise use-cases where veiling may be found to be efficiency enhancing. Here it will suffice to remain with the CFMM case and provide a brief summary of the welfare-improvements from veiling. This applies to the case described above where users have divergent opinions about the resolution of an ambiguous value, but the protocol and/or applications must ultimately resolve to single execution price to remain coherent.

\begin{corollary}[Welfare Improvement for the Entire Type Space $\Theta$]
\label{cor:welfare}
Combining 
Proposition~\ref{prop:no-sure-loss} 
and 
Proposition~\ref{prop:pareto-veil} 
shows that every $\theta \in \Theta$ 
achieve \emph{strictly} higher expected utility under veiling, while the special case
of neutral \(\rho=1\) and \(\alpha=\tfrac12\) users are provided a new pareto dominant equilibrium. Because each $\theta$-agent’s payoff from the veiled design 
is $\ge$ the payoff from the standard design, in many cases strictly so, 
the ex‐ante expected utility $\int_{\Theta}U_\theta\,d\nu(\theta)$ 
is also $\ge$. Hence the ex-ante welfare is (weakly) greater than under the single-probability approach, demonstrating a Pareto improvement for all agent types in $\Theta$.
\end{corollary}

\section{Conclusion}
\label{sec:conclusion}
In summary, we propose that state machines may benefit from enabling non‐measurable ambiguity in (some of) their interactions. The machine can implement veiling, and thereby avoid systematically collapsing set-valued information into single‐probabilities, avoiding the familiar Dutch‐book pitfalls. We demonstrated how such \emph{veiling} can be implemented dynamically and how it may yield fairness and efficiency gains in a blockchain context, with the potential to enhance fairness and efficiency. While much remains to be explored about the practical
implications of incorporating ambiguity on‐chain, our results demonstrate a method by which state machines can preserve liveness even under adversarial conditions, while perhaps also facilitating robust and more inclusive forms of strategic interaction.

\section*{Appendix}
\section{Alternative Functionals for Ambiguous BRL}
\label{app:fairness-functions}
Applications could also try to create other derivative assets that match user preferences, for instance consider the aggregator function from above:

\begin{equation}
   U(w,b)
   \;=\;
   \bigl[\alpha\,w^\rho \;+\; (1-\alpha)\,b^\rho\bigr]^{\tfrac{1}{\rho}}.
\label{eq:OHEU2}
\end{equation}

Where $w$ and $b$ would now denote the two prices, and $\alpha$ and $\rho$ are again parameters that determine the agent's elasticity of substitution between them. But now interpret parameters as reflecting a "fairness attitude":
\begin{itemize}
    \item With $\rho=1$, $0 \le \alpha \le 1$, and $w$ fixed to be one of the two currencies, this provides a linear weighting in the prices. The higher it is, the more the currency $w$ is favored as the true value. 
    \item With $\rho=1$, $\alpha=1$ and $w$ always the lower of the two prices, this always treats the lowest price as the true price. This can be seen as favoring the worst off group.
    \item At $\rho = 0$, we have $w^\alpha b^{1-\alpha}$. This can be seen as a smoothly increasing preference for $w$ when the inequality is extreme. 
\end{itemize}

\section{Additional Details on Veiling Equilibria}
\label{app:VeilingDetails}

We provide here the algebraic details and the best-response verifications 
omitted in the main text. 

\subsection{Corner Payoff Computations}

Recall that Row's payoff is
\[
\pi_1(p,q) 
\;=\; 
p\,\bigl(1-q\bigr)
\;+\;
\lambda\,\bigl(1-p\bigr)\,q
\quad
\text{and}
\quad
e \;=\; \frac{1}{1+\lambda}.
\]
If Row chooses a set $[\,e-a,\,e+a]$ and Column does likewise, the aggregator payoff is 
\[
U_{\tfrac12,1}\bigl([e-a,e+a],[e-a,e+a]\bigr)
\;=\;
\tfrac12\,\min_{p,q\in[e-a,e+a]} \pi_1(p,q)
\;+\;
\tfrac12\,\max_{p,q\in[e-a,e+a]} \pi_1(p,q).
\]
By bilinearity, these extrema occur at corners of the square 
$[\,e-a,e+a]\times[e-a,e+a]$.  A short inspection shows:
\[
\text{(i) Maximum corner: }(p,q) \;=\;(e+a,\,e-a),
\quad
\text{(ii) Minimum corner: }(p,q) \;=\;(e-a,\,e+a),
\]
provided $a$ is small enough that $e-a \ge 0$ and $e+a \le 1$.  
Hence
\[
U_{\tfrac12,1}(A,A)
\;=\;
\tfrac12\,\pi_1(e-a,\,e+a) \;+\; \tfrac12\,\pi_1(e+a,\,e-a).
\]

A direct expansion yields
\begin{align*}
\pi_1(e-a,\,e+a)
&=\;
(e-a)\bigl[1-(e+a)\bigr] \;+\;
\lambda\,\bigl[1-(e-a)\bigr]\,(e+a),
\\
\pi_1(e+a,\,e-a)
&=\;
(e+a)\bigl[1-(e-a)\bigr] \;+\;
\lambda\,\bigl[1-(e+a)\bigr]\,(e-a).
\end{align*}
If we denote these by $f(a)$ and $g(a)$ respectively, then one can check
\[
f(a) + g(a)
\;=\;
2\,\Bigl[\, (1+\lambda)\,e(1-e) \;+\; (1+\lambda)\,a^2 \Bigr].
\]
Recalling that $e(1-e)\,(1+\lambda) = \lambda/(1+\lambda) = c$, we get
\[
\tfrac12\,\bigl[f(a) + g(a)\bigr]
\;=\;
c + (1+\lambda)\,a^2,
\]
hence
\[
U_{\tfrac12,1}\bigl([e-a,e+a],[e-a,e+a]\bigr)
\;=\;
c + (1+\lambda)\,a^2
\;>\;
c,
\]
proving that these intervals strictly improve on the single‐point mix 
$p=e$ from the viewpoint of the aggregator.

\subsection{No Profitable Expansion or Contraction (Boundary Conditions)}

To confirm $[\,e-a,e+a]$ is indeed a \emph{setwise best response}, we must check:

\begin{itemize}
\item
\emph{No profitable contraction:}
Removing a small neighborhood around either boundary, $e-a$ or $e+a$, 
does not strictly raise the aggregator payoff. 
Concretely, 
\[
U_{\tfrac12,1}\bigl([e-a+\delta,e+a],[e-a,e+a]\bigr)
\;\;=\;\;
U_{\tfrac12,1}\bigl([e-a,e+a],[e-a,e+a]\bigr),
\]
for small $\delta>0$, and similarly for $e+a-\delta$.  

\item
\emph{No profitable expansion:}
Including a new probability mass $p<e-a$ or $p>e+a$ in Row’s set does not 
increase the payoff.  Again this is reflected in corner payoffs not 
improving when we push the boundary beyond $e-a$ or $e+a$.
\end{itemize}

Since $\pi_1$ is bilinear, these conditions simplify to linear boundary 
constraints.  One solves them (or checks they are non-strict) to find that 
some nonzero $a^*$ satisfies them exactly, thus pinning down a 
\emph{setwise best response} interval.

\subsection{Opponent’s Best Response and Symmetry}

The game is “asymmetric” only in the sense that Row’s payoff is 
$p(1-q) + \lambda(1-p)q$ while Column’s is $q(1-p) + \lambda(1-q)p$.  
But a straightforward symmetry argument (interchange $\lambda$ and $1$ 
where appropriate) shows that if $[e-a,e+a]$ is a best response for Row, 
then $[e-a,e+a]$ is also a best response for Column.  
Hence $(\,A,A)$ with $A=[e-a,e+a]$ becomes a (symmetric) 
\emph{muddled equilibrium}.

\subsection{Remark on Standard (Point) Responses}

Interestingly, even if one player is not ambiguous, the single‐point $p=e$ remains a 
best response in the symmetric ambiguous equilibrium. Thus even when one player 
“veils” her strategy in an interval, the other may still choose to 
play a single‐point mix, with the same aggregator payoff resulting.  In this sense, we might say that ambiguous users can profitably trade with ambiguous and non-ambiguous users alike.

\bibliography{ref}

\begin{thebibliography}{25}
\providecommand{\natexlab}[1]{#1}
\providecommand{\url}[1]{\texttt{#1}}
\expandafter\ifx\csname urlstyle\endcsname\relax
  \providecommand{\doi}[1]{doi: #1}\else
  \providecommand{\doi}{doi: \begingroup \urlstyle{rm}\Url}\fi

\bibitem[Adams et~al.(2025)Adams, Moallemi, Reynolds, and Robinson]{adams2025amammauctionmanagedautomatedmarket}
Austin Adams, Ciamac~C. Moallemi, Sara Reynolds, and Dan Robinson.
\newblock Am-{{AMM}}: {{An Auction-Managed Automated Market Maker}}.
\newblock Available at \href{https://arxiv.org/abs/2403.03367}{\url{https://arxiv.org/abs/2403.03367}}, February 2025.

\bibitem[Angeris et~al.(2023)Angeris, Chitra, Diamandis, Evans, and Kulkarni]{angeris2023geometry}
Guillermo Angeris, Tarun Chitra, Theo Diamandis, Alex Evans, and Kshitij Kulkarni.
\newblock The geometry of constant function market makers.
\newblock \emph{arXiv preprint arXiv:2308.08066}, 2023.

\bibitem[Arrow(1971)]{arrow1971essays}
K.J. Arrow.
\newblock \emph{Essays in the Theory of Risk-bearing}.
\newblock Markham economics series. North-Holland, 1971.
\newblock ISBN 9780720430479.

\bibitem[Bade(2023)]{badeAmbiguityDesign2023}
Sophie Bade.
\newblock Ambiguity by {{Design}}.
\newblock Available at SSRN: \href{https://ssrn.com/abstract=4561349}{\url{https://ssrn.com/abstract=4561349}}, 2023.

\bibitem[Binmore(2008)]{binmore2008rational}
Ken Binmore.
\newblock Rational decisions.
\newblock In \emph{Rational Decisions}. Princeton University Press, 2008.
\newblock ISBN 9780691149899.

\bibitem[Binmore(2016)]{binmore2016minimal}
Ken Binmore.
\newblock A minimal extension of {{Bayesian}} decision theory.
\newblock \emph{Theory Decis}, 80\penalty0 (3):\penalty0 341--362, March 2016.
\newblock ISSN 0040-5833, 1573-7187.
\newblock \doi{10.1007/s11238-015-9505-0}.

\bibitem[Breidenbach et~al.(2021)Breidenbach, Cachin, Chan, Coventry, Ellis, Juels, Koushanfar, Miller, Magauran, Moroz, et~al.]{breidenbach2021chainlink}
Lorenz Breidenbach, Christian Cachin, Benedict Chan, Alex Coventry, Steve Ellis, Ari Juels, Farinaz Koushanfar, Andrew Miller, Brendan Magauran, Daniel Moroz, et~al.
\newblock Chainlink 2.0: Next steps in the evolution of decentralized oracle networks.
\newblock \emph{Chainlink Labs}, 1:\penalty0 1--136, 2021.

\bibitem[Buterin(2023)]{buterin2023dont}
Vitalik Buterin.
\newblock Don't overload ethereum's consensus, 2023.
\newblock URL \url{https://vitalik.eth.limo/general/2023/05/21/dont_overload.html}.
\newblock Accessed: 2025-03-13.

\bibitem[De~Finetti(1931)]{de1931sul}
Bruno De~Finetti.
\newblock Sul significato soggettivo della probabilittext{\`a}.
\newblock \emph{Fundamenta mathematicae}, 17, 1931.

\bibitem[Drake(2020)]{drake2020enshrined}
Justin Drake.
\newblock Enshrined eth2 price feeds, 2020.
\newblock URL \url{https://ethresear.ch/t/enshrined-eth2-price-feeds}.
\newblock Accessed: 2025-03-13.

\bibitem[Eskandari et~al.(2021)Eskandari, Salehi, Gu, and Clark]{eskandari2021sok}
Shayan Eskandari, Mehdi Salehi, Wanyun~Catherine Gu, and Jeremy Clark.
\newblock Sok: Oracles from the ground truth to market manipulation.
\newblock In \emph{Proceedings of the 3rd ACM Conference on Advances in Financial Technologies}, pages 127--141, 2021.

\bibitem[Gilboa and Schmeidler(1989)]{gilboa1989maxmin}
Itzhak Gilboa and David Schmeidler.
\newblock Maxmin expected utility with non-unique prior.
\newblock \emph{Journal of Mathematical Economics}, 18\penalty0 (2):\penalty0 141--153, January 1989.
\newblock ISSN 0304-4068.
\newblock \doi{10.1016/0304-4068(89)90018-9}.

\bibitem[Graczyk et~al.(2024)Graczyk, Hart, and Mareneck]{graczyk2024slinky}
Marc Graczyk, Sam Hart, and Maghnus Mareneck.
\newblock How slinky can prevent oracle manipulation attacks.
\newblock \emph{Cosmos Ecosystem Blog}, Feb 2024.
\newblock URL \url{https://blog.cosmos.network/in-late-december-an-attack-on-levana-resulted-in-1-146m-2c7da4af54e8}.

\bibitem[Grant et~al.(2023)Grant, Rich, and Stecher]{grant2023decision}
Simon Grant, Patricia Rich, and Jack~Douglas Stecher.
\newblock Decision {{Under Ambiguity}} via {{Intermediate Microeconomics}}, September 2023.
\newblock Available at SSRN: \href{https://ssrn.com/abstract=3369078}{\url{https://ssrn.com/abstract=3369078}}.

\bibitem[Gul and Pesendorfer(2014)]{gul2014expected}
Faruk Gul and Wolfgang Pesendorfer.
\newblock Expected {{Uncertain Utility Theory}}.
\newblock \emph{Econometrica}, 82\penalty0 (1):\penalty0 1--39, 2014.
\newblock ISSN 0012-9682.
\newblock \doi{10.3982/ECTA9188}.

\bibitem[Hurwicz(1951)]{hurwicz1951optimality}
Leonid Hurwicz.
\newblock Optimality criteria for decision making under ignorance.
\newblock Technical report, Cowles Commission discussion paper, statistics, 1951.

\bibitem[Milionis et~al.(2024)Milionis, Moallemi, Roughgarden, and Zhang]{milionis2024automatedmarketmakinglossversusrebalancing}
Jason Milionis, Ciamac~C. Moallemi, Tim Roughgarden, and Anthony~Lee Zhang.
\newblock Automated {{Market Making}} and {{Loss-Versus-Rebalancing}}.
\newblock Available at \href{https://arxiv.org/abs/2208.06046}{\url{https://arxiv.org/abs/2208.06046}}, May 2024.

\bibitem[Peterson et~al.(2019)Peterson, Krug, Zoltu, Williams, and Alexander]{peterson2019augur}
Jack Peterson, Joseph Krug, Micah Zoltu, Austin~K Williams, and Stephanie Alexander.
\newblock Augur: a decentralized oracle and prediction market platform (v2. 0).
\newblock \emph{Whitepaper, https://augur. net/whitepaper. pdf}, 2019.

\bibitem[Pint{\'e}r(2022)]{pinter2022make}
Mikl{\'o}s Pint{\'e}r.
\newblock How to make ambiguous strategies.
\newblock \emph{Journal of Economic Theory}, 202:\penalty0 105459, June 2022.
\newblock ISSN 00220531.
\newblock \doi{10.1016/j.jet.2022.105459}.

\bibitem[Ramsey(2000)]{ramsey2000foundations}
F.P. Ramsey.
\newblock \emph{The Foundations of Mathematics and Other Logical Essays}.
\newblock Number v. 5 in International Library of Philosophy Series. Routledge, 2000.
\newblock ISBN 9780415225465.

\bibitem[Smith(1961)]{smith1961consistency}
Cedric A.~B. Smith.
\newblock Consistency in {{Statistical Inference}} and {{Decision}}.
\newblock \emph{Journal of the Royal Statistical Society Series B: Statistical Methodology}, 23\penalty0 (1):\penalty0 1--25, January 1961.
\newblock ISSN 1369-7412, 1467-9868.
\newblock \doi{10.1111/j.2517-6161.1961.tb00388.x}.

\bibitem[Stecher et~al.(2011)Stecher, Shields, and Dickhaut]{stecher2011generating}
Jack Stecher, Timothy Shields, and John Dickhaut.
\newblock Generating {{Ambiguity}} in the {{Laboratory}}.
\newblock \emph{Management Science}, 57\penalty0 (4):\penalty0 705--712, 2011.
\newblock ISSN 0025-1909.
\newblock \doi{10.1287/mnsc.1100.1307}.

\bibitem[Sutton(2006)]{sutton2006flexibility}
John Sutton.
\newblock Flexibility, {{Profitability}} and {{Survival}} in an ({{Objective}}) {{Model}} of {{Knightian Uncertainty}}.
\newblock Technical report, Working paper, London School of Economics, 2006.

\bibitem[Team(2017)]{team2017gnosis}
Gnosis Team.
\newblock Gnosis-whitepaper.
\newblock \emph{URL: https://gnosis. pm/resources/default/pdf/gnosis whitepaper. pdf}, 2017.

\bibitem[{Wikipedia contributors}(2025)]{wikipedia2025battle}
{Wikipedia contributors}.
\newblock Battle of the sexes (game theory), 2025.
\newblock URL \url{https://en.wikipedia.org/wiki/Battle_of_the_sexes_(game_theory)}.
\newblock Accessed: 2025-03-13.

\end{thebibliography}

\end{document}